%% file: main.tex
\documentclass{article}

\usepackage{arxiv}

\usepackage[pagebackref=true]{hyperref}

\usepackage{graphics} 
\usepackage{epsfig} 
\usepackage{mathptmx} 
\usepackage{times} 
\usepackage{amsmath} 
\usepackage{amssymb}  
\usepackage{amsthm}
\usepackage{xypic}
\usepackage{accents}
\usepackage[ruled,lined]{algorithm2e}

\usepackage[dvipsnames]{xcolor}
\input{preamble_compatible}

\input{notation_defs}

\usepackage{stfloats}

\newcommand{\papercitation}{
Ge, Y., Zamani, B., van Goor, P., Trumpf, J., Mahony, R. (2024), ``Geometric Data Fusion for Collaborative Attitude Estimation'', \emph{Accepted for presentation at 26th International Symposium on Mathematical Theory of Networks and Systems (MTNS)}.
}
\newcommand{\publicationdetails}{
    \papercitation \\
    \copyrightNoticeIFACAccepted{2024} 
}
\newcommand{\publicationversion}
{Author accepted version}

\providecommand{\logG}{\log_\grpG}
\providecommand{\expG}{\exp_\grpG}
\begin{document}

\title{Geometric Data Fusion for Collaborative Attitude Estimation}

\headertitle{
    Geometric Data Fusion for Collaborative Attitude Estimation
    }

    
    
    \author{
Yixiao Ge
\\
    Systems Theory and Robotics Group \\
    School of Engineering \\
	Australian National University \\
    Canberra, Australia \\
    \texttt{Yixiao.Ge@anu.edu.au} \\
    \And    
Behzad Zamani
\\
    Department of Electrical and Electronic Engineering \\
    University of Melbourne \\
    Melbourne, Australia \\
    \texttt{Behzad.Zamani@unimelb.edu.au} \\
\And    
Pieter van Goor
\\
    Robotics and Mechatronics (RaM) Group \\
    EEMCS Faculty\\
    University of Twente \\
    Enschede, The Netherlands \\
    \texttt{p.c.h.vangoor@utwente.nl} \\
	\And	
    Jochen Trumpf
\\
    School of Engineering \\
	Australian National University \\
    Canberra, Australia \\
	\texttt{Jochen.Trumpf@anu.edu.au} \\
    \And	
    Robert Mahony
\\
    Systems Theory and Robotics Group \\
    School of Engineering \\
	Australian National University \\
    Canberra, Australia \\
	\texttt{Robert.Mahony@anu.edu.au} \\
}

\maketitle



\begin{abstract}                
    In this paper, we consider the collaborative attitude estimation problem for a multi-agent system.
    The agents are equipped with sensors that provide directional measurements and relative attitude measurements.
    We present a bottom-up approach where each agent runs an extended Kalman filter (EKF) locally using directional measurements and augments this with relative attitude measurements provided by neighbouring agents. 
    The covariance estimates of the relative attitude measurements are geometrically corrected to compensate for relative attitude between the agent that makes the measurement and the agent that uses the measurement before being fused with the local estimate using the convex combination ellipsoid (CCE) method to avoid data incest. 
    Simulations are undertaken to numerically evaluate the performance of the proposed algorithm. 
\end{abstract}


\section{Introduction}
The problem of collaborative state estimation over sensor networks has drawn significant attention in the past 20 years \cite{roumeliotis2002distributed}.
In this problem, different agents share measurements and state-estimates to improve overall state estimation.
Sharing data in this way, however, introduces the possibility of data incest \cite{julier1997non}.
To see this, consider a network of individual estimators each estimating some states while communicating with other nodes on the network.
Information received from other agents will depend on information transmitted by the agent itself in preceding communications, potentially reinforcing its own hypothesis and increasing the risk of overconfidence in the resulting state estimates \cite{julier1997non}.

To overcome these challenges there are two main approaches: a top down approach where the full state estimation is formulated as a joint estimation problem and then the computation is decentralised to each node (\cite{carrillo2013decentralized, luft2018recursive,zhu2020fully}), and the bottom up approach, where each agents runs an independent estimator locally and fuses data from other agents taking precautions to avoid data incest \cite{olfati2005distributed,olfati2006distributed}.
The key enabling step in the bottom up approach is a methodology to provide safe fusion of correlated data into a local agent state estimation such that it retains the common uncertainty of the original random variables.
This problem has been studied since the 60s \cite{schweppe1968recursive,kahan1968circumscribing, bertsekas1971recursive}.
In more recent work (\cite{julier1997non}, \cite{julier2017general}) Julier and Uhlmann proposed the Covariance Intersection (CI) algorithm which restricts the fusion problem to a family of convex combinations of the inverse covariance matrices and is the most commonly used data fusion method in multi-agent problems.
The CI algorithm, however, is known to be too conservative in certain situations \cite{julier2007using}.
The Inverse Covariance Intersection (ICI) method \cite{noack2017decentralized} computes the maximum possible common information shared by the estimates to be fused, and is known to be less conservative than the CI method.
An alternative solution is the Convex Combination Ellipsoid (CCE) fusion method which arises from the set-theoretic fusion technique \cite{schweppe1968recursive}.
The CCE method shares a similar structure with CI, however it improves the tightness of the fusion result while avoiding unnecessary uncertainties as the byproduct of the fusion process \cite{zamani2023collaborative}.

All these fusion algorithms are originally formulated in global Euclidean space, and there have been many attempts to adapt these classical methods to systems that live on smooth manifolds, particularly Lie groups.
One popular approach is to consider the fusion problem as finding the optimal mode of the posterior distribution by solving an optimization problem \cite{bourmaud2016intrinsic}.
In \cite{wolfe2011bayesian}, the authors solved the fusion problem by posing a set of algebraic equations using the Baker-Campbell-Hausdorff formula.
Recent work in equivariant filter theory \cite{mahony2020equivariant} and geometric extended Kalman filtering \cite{mueller2017covariance,ge2023note} has provided a strong understanding of the geometry of filtering and data fusion.
In particular, these works demonstrate that it matters in which coordinates the generative noise model for a measurement is posed and provides formulae and methodology to transform covariance from one set of coordinates to another \cite{ge2023note}.



In this paper, we consider a bottom up approach to the problem of collaborative attitude estimation, where each node estimates its own attitude as well as taking relative measurements of other nodes.
The problem is posed on the Lie group $\SO(3)$ representing the attitude of a single agent of interest, termed the \emph{ego-agent}. 
The information used are local directional measurements, angular velocity, and a noisy relative attitude measurement of the ego-agent as observed by a neighbouring \emph{altruistic-agent} along with the altruistic agent's own state estimate (estimated attitude and its error covariance).
This relative attitude measurement, combined with the altruistic agent's state estimate, is effectively an attitude measurement of the ego-agent, and can be fused into the ego-agent's state estimation, at the appropriate point, in the filter algorithm.
However, in a collaborative estimation scenario, the altruistic agent's state estimate is itself dependent on shared information from the ego-agent, and this relative pose measurement should not be treated as an independent measurement.
Furthermore, the attitude measurement is observed from the perspective of the altruistic agent and is written in these coordinates.
The covariance of the measurement must be transformed into the ego-state coordinates to avoid incorrect inference.
The contribution of the paper is to combine the geometric modification of the covariance into the correct coordinates with the CCE fusion method to obtain a high-performance bottom up collaborative state estimation scheme for multi-agent attitude estimation.
We provide a Monte-Carlo simulation to demonstrate the performance of the proposed estimation algorithm together with the geometric modifications.


\section{Preliminaries}




\subsection{Special orthogonal group $\SO(3)$}
Attitude of an agent is represented as a rotation matrix $R$ in the special orthogonal group $R \in \SO(3)$. 
The identity element of $\SO(3)$, denoted $\id$, is the identity matrix. 
Given arbitrary $X,Y\in\SO(3)$, the left and right translations are denoted by $\textrm{L}_X(Y) :=XY$ and $\textrm{R}_X(Y):=YX$.
The Lie algebra $\so(3)$ of $\SO(3)$ consists of all skew-symmetric matrices of the form
\[
u^\wedge=\left(\begin{array}{ccc}
0 & -u_3 & u_2 \\
u_3 & 0 & -u_1 \\
-u_2 & u_1 & 0
\end{array}\right),
\]
and is isomorphic to the vector space $\R^3$. 
We use the wedge $(\cdot)^\wedge:\R^3\rightarrow\so(3)$ and vee $(\cdot)^\vee:\so(3)\rightarrow\R^3$ operators to map between the Lie algebra and vector space.
The Adjoint map for the group $\SO(3)$, $\Ad_X:{\so(3)}\to{\so(3)}$ is defined by
\[
    \Ad_{X}[{{u}^{\wedge}}] = X u^\wedge X^\top  ,
\]
for every $X \in \SO(3)$ and ${{u}^{\wedge} \in \so(3)}$.
Given particular wedge and vee maps, the Adjoint matrix is defined as the map ${\Ad_{X}^\vee: {\R^{3}}\to{\R^{3}}}$
\begin{equation*}
    \Ad_{X}^\vee u = \left(\Ad_{X}{{u}^{\wedge}}\right)^{\vee} .
\end{equation*}
The adjoint map for the Lie algebra $\ad_{{u}^\wedge}: {\so(3)}\to{\so(3)}$ is given by
\begin{equation*}
    \ad_{{u}^\wedge}{{v}^{\wedge}} = \left[{u}^{\wedge}, {v}^{\wedge}\right].
\end{equation*}
We define the adjoint matrix $\ad^\vee_u:\R^3\to\R^3$ to be:
\begin{align*}
    \ad^\vee_u v=\left[{u}^{\wedge}, {v}^{\wedge}\right]^\vee.
\end{align*}

Let $\expG:\so(3)\to\SO(3)$ denote the matrix exponential (G denotes the SO(3) group).
In order to improve the analogy to the fusion literature that is usually written in $\R^n$ coordinates we will use the $\boxplus$ (`boxplus') operator for the exponential map
\[
X\boxplus u = X\expG(u^\wedge),
\]
for $X\in\SO(3)$ to represent the state and $u\in\mathbb{R}^3$ to represent a certain noise process. 
Let $\SO^\circ(3)\subset\SO(3)$ be the subset of $\SO(3)$ where the exponential map is invertible and note that $\SO^\circ(3)$ is almost all of $\SO(3)$, excluding only those points with a rotation of $\pi$ radians. 
The logarithm map $\logG:\SO^\circ(3)\to\so(3)$ and $\logG^\vee:\SO^\circ(3)\to\R^3$ is well defined on $\SO^\circ(3)$. 

The Jacobian $\textbf{J}_{u^\wedge}:\so(3)\to\so(3)$ is defined to be the left-trivialised directional derivative of $\expG : \so(3) \to \SO(3)$ at a point $u^\wedge\in\so(3)$ on $\SO(3)$.
Given an arbitrary $w^\wedge\in\so(3)$, it satisfies
\begin{align*}
    \textbf{J}_{u^\wedge}[w^\wedge] &= \tD\tL_{\expG(-u^\wedge)}\cdot\tD \exp_{\grpG}(u^\wedge)[w^\wedge],
\end{align*}
where the tangent space $\tT_{\exp_{\grpG}(u^\wedge)} \SO(3)$ is isomorphic to $\so(3)$ via left trivialisation.
Equivalently, $\tD \exp_{\grpG} (u^\wedge) [w^\wedge] =  \tD\tL_{\exp_{\grpG}(u^\wedge)}\textbf{J}_{u^\wedge} [w^\wedge] \in \tT_{\exp_{\grpG}(u^\wedge)} \SO(3)$. 
Its matrix form, denoted by $J_u\in\R^{3\times 3}$, is given by \cite{Chirikjian_2011}
\begin{align*}
J_u &:=I_3-\frac{1-\cos \|u\|}{\|u\|^2} u^\wedge+\frac{\|u\|-\sin \|u\|}{\|u\|^3} {u^\wedge}^2.
\end{align*}
For an arbitrary $u^\wedge\in\so(3)$, the inverse of its Jacobian is given by 
\begin{align*}
    J_u^{-1}:=I_3+\frac{1}{2} u^\wedge+\left(\frac{1}{\|u\|^2}-\frac{1+\cos\|u\|}{2\|u\|\sin\|u\|}\right){u^\wedge}^2.
\end{align*}


\subsection{Concentrated Gaussians on $\SO(3)$}
We use the concept of a concentrated Gaussian to model distributions on $\SO(3)$.
For a random variable $X\in\SO(3)$, the probability density function is defined as
\begin{align}\label{eq:CGD}
    p(X;\hat{X},\mu,\Sigma) = \alpha e^{-\frac{1}{2}(\logG^\vee(\hat{X}^{-1}X)-\mu)^\top\Sigma^{-1}(\logG^\vee(\hat{X}^{-1}X)-\mu)},
\end{align}
where $\alpha$ is a normalizing factor.
The stochastic parameters are $\mu\in\R^3$, a mean vector in local coordinates, and $\Sigma\in\mathbb{S}_+(3)$ a positive-definite symmetric $3\times 3$ covariance matrix parameter.
The geometric parameter $\hat{X}\in\SO(3)$ is termed the reference point and plays the role of the origin of the local coordinates.
We will term a concentrated Gaussian \emph{zero-mean} if $\mu \equiv 0$. 
In this case the distribution corresponds to the classical concentrated Gaussian \cite{chirikjian2014gaussian,bourmaud2016intrinsic} where one can think of the reference point $\hat{X} \in \SO(3)$ as a sort of `geometric' mean. 
We will write $X\sim\GP_{\hat{X}}(\mu,\Sigma)$ for the random variable $X \in \SO(3)$. 

\begin{lemma}\label{lemma:coordinate_change}
    Given an arbitrary concentrated Gaussian distribution $p(X) = \GP_{X_1}(\mu_1,\Sigma_1)$ on $\SO(3)$, then the zero-mean concentrated Gaussian distribution $q(X) = \GP_{X_2}(0,\Sigma_2)$ with parameters
    \begin{align*}
        &X_2 = X_1\expG(\mu_1)\\
        &\Sigma_2 = J_{\mu_1}\Sigma_1 J_{\mu_1}^\top
    \end{align*}
    minimises the Kullback-Leibler divergence $p(X)$ with respect to $q(X)$ up to second-order linearisation error.
\end{lemma}
\begin{proof}
The Kullback-Leibler divergence between $p(X)$ and $q(X)$ is given by
\begin{align*}
        {\text{KL}}(p||q)&= \E_p[\log(p)-\log(q)]\\
            & = C_p +\frac{n}{2}\log(2\pi)+ \frac{1}{2}\log\det(\Sigma_2)\\
            &\qquad\qquad+\frac{1}{2}\E_p[\logG^\vee(X_2^{-1}X)^\top \Sigma_2^{-1}\logG^\vee(X_2^{-1}X)],
    \end{align*}
    where $C_p$ is the negative entropy of $p(X)$.
    Taking the derivative of $ {\text{KL}}(p||q)$ with respect to $\Sigma_2$ in the direction $u$ yields
    \begin{align*}
        \tD&_{\Sigma_2}\text{KL}(p(X)||q(X))[u] = \\
        &\frac{1}{2}\tr\big(\Sigma_2^{-1}u-
        \Sigma_2^{-1} \E_p[\logG^\vee(X_2^{-1}X)\logG^\vee(X_2^{-1}X)^\top] \Sigma_2^{-1}u\big).
    \end{align*}
    The critical point is given by
     \begin{align}\label{eq:R_prime_expect}
        \Sigma_2 = \E_p[\logG^\vee(X_2^{-1}X)\logG^\vee(X_2^{-1}X)^\top].
    \end{align}

    Defining $\phi_1:\R^3\to\R^3$ and $\phi_2:\R^3\to\R^3$ as
    \begin{align*}
        &\phi_1(X): = \logG^\vee(X_1^{-1}X)-\mu_1,\\
        &\phi_2(X): = \logG^\vee(\expG(-\mu_1)X_1^{-1}X),
    \end{align*}
    one has
    \begin{align}\label{eq:transfer_function}
        \phi_2(X) = \logG^\vee(\expG(\mu_1)^{-1}\expG(\phi_1(X)+\mu_1)).
    \end{align}
    Taking the Taylor series of \eqref{eq:transfer_function} at $\phi_1(X) = 0$ up to first order yields:
    \begin{align*}
        \phi_2(X) &\approx \tD\logG^\vee(\id)\circ\tD \textrm{L}_{\expG(\mu_1)^{-1}}\circ\tD \expG(\mu_1) [\phi_1(X)]\\
        &=J_{\mu_1}\phi_1(X).
    \end{align*}
    Substitute the result into \eqref{eq:R_prime_expect}:
    \begin{align*}
        \Sigma_2 &= \E_p[\phi_2(X)\phi_2(X)^\top]\\
        &\approx\E_p[J_{\mu_1}\phi_1(X)(J_{\mu_1}\phi_1(X))^\top]\\
        &=J_{\mu_1}\E_p[\phi_1(X)\phi_1(X)^\top]J_{\mu_1}^\top\\
        &=J_{\mu_1}\Sigma_1J_{\mu_1}^\top,
    \end{align*}
    where the last equality follows from the definition of $\Sigma_1 =\E_p[\phi_1(X)\phi_1(X)^\top] $.
\end{proof}

\begin{corollary}\label{corollary:inverse}
    Given an arbitrary zero-mean concentrated Gaussian distribution $p(X) = \GP_{X_1}(0,\Sigma_1)$, choose and fix $X_2\in\SO(3)$, then the concentrated Gaussian $q(X) = \GP_{X_2}(\mu_2,\Sigma_2)$ for parameters
    \begin{align*}
        &\mu_2 = \logG^\vee(X_2^{-1}X_1)\\
        &\Sigma_2 = J_{\mu_2}^{-1}\Sigma_1J_{\mu_2}^{-\top}
    \end{align*}
    minimises the Kullback-Leibler divergence with $p(X)$ up to second-order linearisation error.
\end{corollary}

Corollary \ref{corollary:inverse} follows directly from Lemma \ref{lemma:coordinate_change}.

\section{Problem Formulation}
The problem we target is to design a fully distributed algorithm to estimate the absolute attitude of individual agents collaboratively in a multi-agent system.
Each agent is equipped with an onboard inertial measurement unit (IMU) that provides bias-free angular velocity $\omega\in\R^3$ in the body frame.
With non-rotating, flat earth assumption, the deterministic system kinematics are given by
\begin{align}\label{eq:attitude}
   R_i(t+1) = R_i(t) \exp_\grpG \left( \Delta t\;\omega_i(t)^{\wedge}\right).
\end{align}
Note that the subscripted index $i$ refers to the $i^\text{th}$ agent and the index $t$ refers to the time step. 
We will drop the time step notation where it is clear in order to simplify notation. 

The agent $i$ also has extrinsic sensors such as a magnetometer that measures known directions (magnetic field) in the body-frame.
The $\ell$th direction measurement $\idx{z}{i}{}{\ell}$ for agent $i$ is given by 
\begin{align}\label{eq:output}
\idx{z}{i}{}{\ell} = R_i^\top d_\ell
   \end{align}
where $d_\ell$ for $\ell = 1, \ldots, n$ are a collection of known reference directions. 
Using these measurements, the agent can run a filter-based algorithm locally to estimate its own state and the associated covariance \cite{markley2014fundamentals}\cite{fornasier2022overcoming}.

In addition, each agent can communicate with agents in its neighbourhood and has a sensor capable of measuring relative attitude of its neighbouring agents.
If two agents $i$ and $j$ have states $R_i, R_j \in \SO(3)$ then the relative state $\idx{R}{j}{}{i}$ of $j$ with respect to $i$ is defined to be
\[
\idx{R}{j}{}{i} := R_j^{-1} R_i \in \SO(3).
\]
The relative state can be thought of as the attitude of agent $i$ expressed as perceived by a sensor on agent $j$. 
Alternatively, consider left translation of the whole space by  $\tL_{R_j^{-1}}$.
Then $\id = \tL_{R_j^{-1}} R_j = R_j^{-1} R_j = \Id_3$ and $\idx{R}{j}{}{i} = \tL_{R_j^{-1}} R_i = R_j^{-1} R_i$.
That is, the relative state is the coordinates of agent $i$ with respect to a new group parametrisation that places agent $j$ at the identity attitude. 

A relative state sensor on agent $j$ may directly measure the physical relative attitude of agent $i$ or may measure the relative angular difference between the states. 
Physical attitude measurements are associated with directly measuring the direction cosines that make up the entries of the matrix $\idx{R}{j}{}{i} \in \SO(3)$. 
Such measurements are typically inner products like $\idx{z}{i}{}{\ell}^\top d_\ell$ that correspond to cosines of angles between known or measured vectors.
For a physical relative attitude measurement, an appropriate model is 
\begin{align}\label{eq:measurement_1}
\idx{y}{j}{}{i} = R_j^{-1}R_i\boxplus\kappa_j,\quad \kappa_j\sim \GP(0,Q_j).
\end{align}
That is, the generative measurement noise model is a zero-mean concentrated Gaussian process $\idx{y}{j}{}{i} \sim \GP_{\idx{R}{j}{}{i}} (0,Q_j)$. 
Conversely, a relative angular sensor measures the underlying angle from one attitude to another. 
For example, if two attitudes are connected through a physical gimbal system then the sensor will measure Euler angles between the two states. 
Another example is when a vision system or similar system estimates the relative axis of rotation and relative angle from itself to another agent rather than the direction cosines \cite{chatterjiVisionBasedPositionAttitude1998}.  
The measurement in this case, which we denote by $z$ to distinguish it from \eqref{eq:measurement_1}, is appropriately modelled by 
\begin{align}\label{eq:measurementmodel_2}
    \idx{z}{j}{}{i} = \expG \left( \idx{\mu}{j}{}{i} + \kappa_j ^{\wedge}\right) ,\quad \kappa_j\sim\GP(0,Q_j)
\end{align}
where $\idx{\mu}{j}{}{i} = \logG(R_j^{-1}R_i) \in \so(3)$ is the angle-axis representation.
$\idx{\mu}{j}{}{i} = \theta a^\wedge$ for a rotation of $\theta$ rad around an axis $a \in \Sph^2$. 
In this case, the generative noise measurement model is a concentrated Gaussian $\idx{z}{j}{}{i} \sim \GP_{\id}(\idx{\mu}{j}{}{i},Q_j)$ with non-zero mean. 

\begin{figure}[htb!]
\centering
    \includegraphics[width=0.7\linewidth]{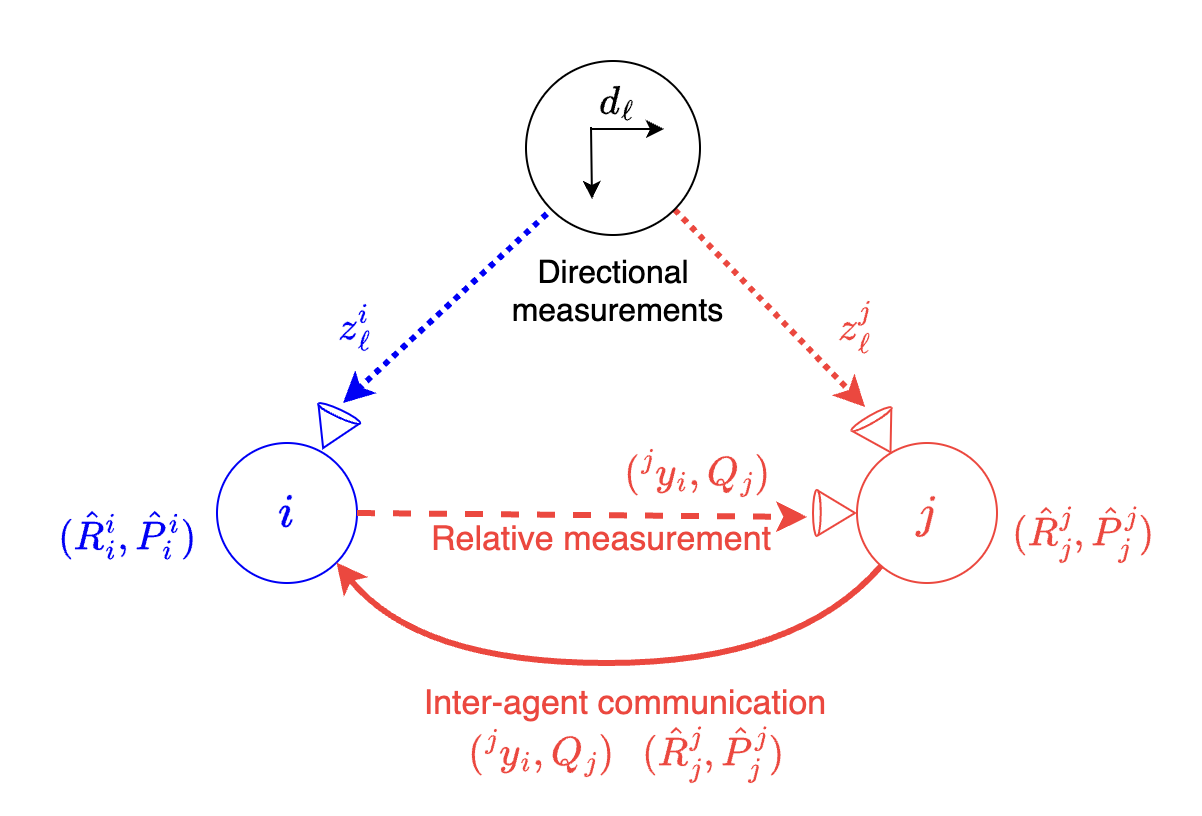}
    \caption{Illustration of the experimental setup to improve the estimate of agent $i$. The dotted lines represent each agent's measurements of known directions which are used to locally estimate their own states. The dashed line refers to the inter-agent measurement ${}^j y_i$ taken by agent $j$. The solid line represents the communication from agent $j$ to agent $i$.
    }
    \label{fig:diagram}
\end{figure}

\begin{figure}
    \centering
    \includegraphics[width=0.7\linewidth]{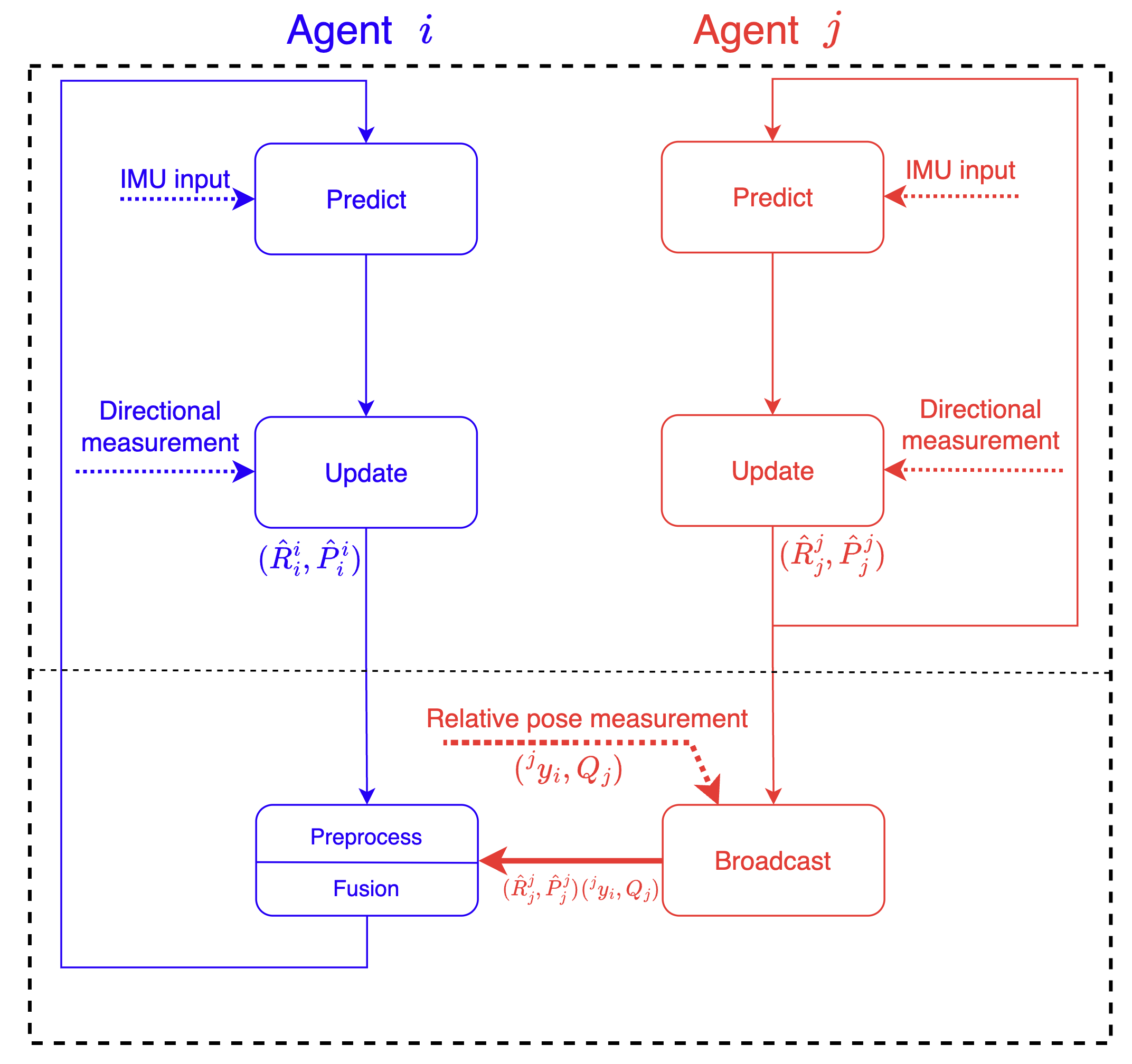}
    \caption{Implementation structure of the estimators and the information flow on each agent.}
    \label{fig:schematic}
\end{figure}

\section{Algorithm}\label{sec:algorithm}
We propose a filter-based algorithm to solve the attitude estimation problem.
Each agent estimates its own state $\hat{R}\in\SO(3)$ and the associated covariance $\hat{P}\in\mathbb{S}_+(3)$.
We use subscript and superscript for the agent that is being estimated and the agent that is making the estimation, respectively.
Fig~\ref{fig:schematic} demonstrates the overview of different steps in the proposed filter algorithm for an ego-agent (shown in blue) and an altruistic agent (shown in red).
In this section, we focus on the details of the filter running on the ego-agent $i$, which can be separated into three stages: predict (using IMU input), update (using directional measurement) 
and fusion (using relative measurements).

\subsection{Predict and update}
The filter follows the conventional EKF design methodology, including the predict and update step. 
The information state of the filter, which approximates the true distribution of the system state on $\SO(3)$, is a concentrated Gaussian distribution, given by
\begin{align*}
    R_i(t) \sim \GP_{\hat{R}_i^i(t)}(0,\hat{P}^i_i(t)).
\end{align*}

When an IMU input is available, it will be used to propagate the state estimate $\hat{R}_i^i$, which also acts as the reference point in the information state of the filter, using the full nonlinear model \eqref{eq:attitude}.
The covariance estimate $\hat{P}^i_i$ will be propagated by linearising the system model.
When the agent receives the directional measurement, under the assumption that every extrinsic measurement is independent, we can perform the standard Kalman fusion in logarithmic coordinates, followed by a covariance reset step which transforms the posterior into a zero-mean concentrated Gaussian (\cite{mueller2017covariance,ge2023note}).
For detailed implementation of the filter, see \cite{fornasier2022overcoming, ge2023note}.

\subsection{Fusion using relative measurements:}
\vspace{-2mm}
The core contribution of this work lies in the preprocessing and fusion steps shown in Figure~\ref{fig:schematic}.
When agent $j$ makes a relative attitude measurement $\idx{y}{j}{}{i}$ or $\idx{z}{j}{}{i}$ of agent $i$, it broadcasts the measurement and the associated noise covariance $Q_j$ of the measurement model, as well as agent $j$'s own state estimate $(\hat{R}^j_j, \hat{P}^j_j)$ to agent $i$.
As demonstrated in Fig \ref{fig:schematic}, it takes two steps to fuse the inter-agent information with agent $i$'s own estimate. 
Firstly in the preprocessing stage, the relative measurement is combined with both agent $i$ and agent $j$'s state estimate, which generates a new estimate of agent $i$, denoted by $(\hat{R}^j_i, \hat{P}^j_i)$.
We use the superscript $j$ to distinguish it from agent $i$'s estimate.
Then we implement the CCE fusion method to fuse $(\hat{R}^j_i, \hat{P}^j_i)$ with $(\hat{R}^i_i, \hat{P}^i_i)$.

\vspace{-1mm}
\subsubsection{Geometric transformation of the measurement model:}
We consider both measurement models in \eqref{eq:measurement_1} and \eqref{eq:measurementmodel_2}.
The second measurement model requires an extra step to transform it into a zero-mean concentrated Gaussian.
As given in \eqref{eq:measurementmodel_2}, the relative measurement can be expressed as a concentrated Gaussian with non-zero mean,
\[
    \idx{z}{j}{}{i}\sim\GP_\id(\logG(R_j^{-1}R_i), Q_j).
\]
By applying Lemma \ref{lemma:coordinate_change}, one gets the following approximation:
\[
\idx{y}{j}{}{i}\sim \GP_{R_j^{-1}R_i}(0, J_{\logG^\vee(R_j^{-1}R_i)}Q_j J_{\logG^\vee(R_j^{-1}R_i)}^\top).
\]
Such modification requires the knowledge of the noiseless configuration output $R_j^{-1}R_i$ which is not available in practice.
In this paper, we will assume the measurement noise is small and use the measurement $\idx{z}{j}{}{i}$ as a proxy for the relative state.
Alternative algorithms that exploit the two state estimates $\hat{R}^j_j$ and $\hat{R}^i_i$ directly are also possible. 
The measurement model can now be transformed as follow:
\[
\idx{y}{j}{}{i}\sim \GP_{R_j^{-1}R_i}(0, J_{\logG^\vee(\idx{z}{j}{}{i})} Q_j J_{\logG^\vee(\idx{z}{j}{}{i})}^\top),
\]
or equivalently,
\begin{align}\label{eq:measurement_j_new}
    \idx{y}{j}{}{i} \approx R_j^{-1}R_i\boxplus\kappa_j,\quad \kappa_j\sim\GP(0,Q^*_j),
\end{align}
with $Q^*_j = J_{\logG^\vee(\idx{y}{j}{}{i})} Q_j J_{\logG^\vee(\idx{y}{j}{}{i})}^\top$.

\subsubsection{Preprocessing relative measurements:}
Given the measurement model
\begin{align}\label{eq:measurement_j}
    \idx{y}{j}{}{i} = R_j^{-1}R_i\boxplus\kappa_j,\quad \kappa_j\sim\GP(0,Q_j)
\end{align}
and the error state models
\begin{align}\label{eq:errorstate_i}
    R_i = \hat{R}_i\boxplus\varepsilon_i,\quad \varepsilon_i\sim\GP(0,\hat{P}_i)
\end{align}
\begin{align}\label{eq:errorstate_j}
    R_j = \hat{R}_j\boxplus\varepsilon_j,\quad \varepsilon_j\sim\GP(0,\hat{P}_j)
\end{align}
substituting \eqref{eq:errorstate_j} into \eqref{eq:measurement_j} yields
\begin{align*}
    \idx{y}{j}{}{i} &= (\hat{R}_j\boxplus\varepsilon_j)^{-1}R_i\boxplus\kappa_j= \expG(-\varepsilon_j)\hat{R}_j^{-1}R_i\boxplus\kappa_j\\
            &= \hat{R}_j^{-1}R_i\boxplus\Ad_{(\hat{R}_j^{-1}R_i)^{-1}}(-\varepsilon_j)\boxplus\kappa_j.
\end{align*}
Replace $R_i$ using \eqref{eq:errorstate_i},
\begin{align*}
    \idx{y}{j}{}{i} &= \hat{R}_j^{-1}R_i\boxplus\Ad_{(\hat{R}_j^{-1}(\hat{R}_i\boxplus\varepsilon_i))^{-1}}(-\varepsilon_j)\boxplus\kappa_j\\
    &=\hat{R}_j^{-1}R_i\boxplus\Ad_{\expG(-\varepsilon_i)}\Ad_{(\hat{R}_j^{-1}\hat{R}_i)^{-1}}(-\varepsilon_j)\boxplus\kappa_j.
\end{align*}
Take the Taylor expansion of $\Ad^\vee_{\expG(-\varepsilon_i)}$, one has
\begin{align*}
    \Ad^\vee_{\expG(-\varepsilon_i)} =  I_3 - \ad^\vee_{\varepsilon_i} + \Order(\lvert \varepsilon_i \rvert^2).
\end{align*}
Assume that both $\varepsilon_i$ and $\varepsilon_j$ are small, then $\ad_{\varepsilon_i}(\varepsilon_j)$ and the higher-order terms can be approximated to be zero:
\begin{align}
    \idx{y}{j}{}{i} &= \hat{R}_j^{-1}R_i\boxplus\Ad_{(\hat{R}_j^{-1}\hat{R}_i)^{-1}}(-\varepsilon_j)\boxplus\kappa_j\nonumber\\
    &\approx \hat{R}_j^{-1}R_i\boxplus\kappa^+_j, \label{eq:new_meas}
\end{align}
where the new measurement noise $\kappa^+_j$ is given by
\begin{align*}
    \kappa^+_j \sim \GP(0,\Ad^\vee_{(\hat{R}_j^{-1}\hat{R}_i)^{-1}}\hat{P}^j_j{\Ad^\vee_{(\hat{R}_j^{-1}\hat{R}_i)^{-1}}}^\top + Q_j).
\end{align*}
One can now reconstruct a new estimate of $R_i$ from \eqref{eq:new_meas} by left multiplying by $\hat{R}_j$.
The new estimate is a zero mean Gaussian $\GP_{\hat{R}^j_i}(0, \hat{P}^j_i)$ where the parameters are given by
\begin{align*}
        \hat{R}^j_i &= \hat{R}_j\idx{y}{j}{}{i}\\
        \hat{P}^j_i &= \Ad^\vee_{(\hat{R}_j^{-1}\hat{R}_i)^{-1}}\hat{P}^j_j{\Ad^\vee_{(\hat{R}_j^{-1}\hat{R}_i)^{-1}}}^\top + Q_j.
\end{align*}

\vspace{-3mm}
\subsubsection{Geometric correction of distributions:}
The concentrated Gaussian distributions that are being fused can be written as $\GP_{\hat{R}^i_i}(0,\hat{P}^i_i)$ and $\GP_{\hat{R}^j_i}(0,\hat{P}^j_i)$.
Although the covariance $\hat{P}^i_i$ and $\hat{P}^j_i$ are both symmetric matrices, they are still associated with distributions expressed in different coordinates.
Fusing these covariance matrices directly, without correcting for the associated change of coordinates, will introduce artifacts and errors in the information state of the resulting filter, decreasing consistency and compromising performance of the algorithm. 

In order for the ego-agent to compensate for the change of coordinates in the measurement recieved from the altruistic agent, it must transform the measurement concentrated Gaussian into a concentrated Gaussian in its local coordinates. 
That is, it must solve for $\mu^{j*}_i$ and $\hat{P}^{j*}_i$ such that 
\begin{align*}
    \GP_{\hat{R}_i^j}(0,\hat{P}^j_i) \approx \GP_{\hat{R}^i_i}(\mu^{j*}_i,\hat{P}^{j*}_i).
\end{align*}
Applying Corollary \ref{corollary:inverse} then one has 
\begin{align*}
\mu^{j*}_i  :=\logG^\vee({\hat{R}^i_i}^{-1}\hat{R}_j\idx{y}{j}{}{i})\; , \qquad
\hat{P}^{j*}_i = J_{\mu^{j*}_i }^{-1} \hat{P}^j_i J_{\mu^{j*}_i }^{-\top}. 
\end{align*}

\vspace{-3mm}
\subsubsection{Data Fusion}
Now the targeting distributions are transported into the same coordinate, and the next step is to perform data fusion to the two distributions.
Rewrite the distributions into ellipsoidal sets on $\so(3)$, defined as $\mathcal{E}(0,\hat{P}^i_i) = \{u^\wedge\in\so(3) : \lVert u \rVert^2_{{\hat{P}^i_i}^{-1}}\leq 1\}$ and $\mathcal{E}(\mu^{j*}_i,\hat{P}^{j*}_i) = \{u^\wedge\in\so(3) : \lVert u - \mu^{j*}_i \rVert^2_{{\hat{P}^{j*}_i}^{-1}}\leq 1\}$.
Given these two prior ellipsoids have nonempty intersection, the convex combination $\mathcal{E}(\hat{u}^+, \hat{P}^+_i)$ of them is given by\small
\begin{align*}
& \hat{P}^{+}_i=k X, \;X=\left(\alpha {\hat{P}^i_i}^{-1}+(1-\alpha) {\hat{P}^{j*}_i}^{-1}\right)^{-1}, \\
& k=1-d^2, \;d^2=\left\|\mu^{j*}_i\right\|_{\left(\frac{\hat{P}_i^i}{\alpha}+\frac{\hat{P}^{j*}_i}{1-\alpha}\right)^{-1}}^2, \\
& \hat{u}^{+}=X\left((1-\alpha) {\hat{P}_i^{j*}}^{-1} \mu^{j*}_i\right) ,
\end{align*}
\normalsize where $\alpha\in[0,1]$ is a freely chosen gain in this paper. Alternatively, one can choose an optimal $\alpha^*$ such that $\alpha^*=\argmin_{\alpha}\det (\hat{P}_i^*)$.

Note that given the outputs of the CCE fusion method, the posterior is a concentrated Gaussian with non-zero mean, that is, $R_i\sim\GP_{\hat{R}_i}(\hat{u}^+, \hat{P}^+_i)$.
However, the next fusion iteration requires the distribution to have a zero mean, hence the goal of the reset step is to identify $\hat{P}^{\diamond}_i$ such that
\begin{align}
    R_i\sim\GP_{\hat{R}_i}(\hat{u}^+, \hat{P}^+_i)\approx\GP_{\hat{R}_i\expG{\hat{u}^+}}(0, \hat{P}^\diamond_i).
\end{align}
Similar to the coordinate change in the previous steps, this may be solved by using Lemma \ref{lemma:coordinate_change}.
The reset covariance $\hat{P}^\diamond_i$ is found to be
$
\hat{P}^\diamond_i = J_{\hat{u}^+}\hat{P}^+_i J_{\hat{u}^+}^\top.
$

Note that the reset step only modifies the covariance estimate and does not change the attitude estimate of agent $i$.

\section{Numerical experiment}
\vspace{-3mm}
In this section, we provide a numerical evaluation of the algorithm proposed in Section \ref{sec:algorithm}.
A Monte-Carlo simulation is undertaken to validate the performance of both the proposed geometric modifications and the fusion algorithm.
\vspace{-3mm}
\subsection{System implementation}
\vspace{-3mm}
In the Monte-Carlo setup we use the following randomisation and run 1000 experiments.
We simulate two independent oscillatory trajectories for agent $i$ and $j$, with the noise-less angular velocity generated by 
\begin{align*}
    \omega_i(\tau): &= (10 \lvert\sin(\tau)\rvert, \lvert\cos(\tau)\rvert, 0.1\lvert\sin(\tau)\rvert)\,\text{rad/s}\\
    \omega_j(\tau): &= (\lvert\sin(\tau)\rvert, 0.5\lvert\cos(\tau)\rvert, 5\lvert\sin(\tau)\rvert)\,\text{rad/s},
\end{align*}
and subsequently corrupted by additive Gaussian noise with covariance $\text{diag}(0.3^2,0.2^2,0.1^2)$.
The trajectory is realized using Euler integration \eqref{eq:attitude} and a time step $\Delta t=0.02s$.
The initial states estimates of the agents are offset from each other by 180 degrees with initial errors sampled from $\GP_{R(0)}(0, I_3)$.

The extrinsic sensor on agent $j$ measures two known directions
$
    d_1 = (0,1,0)\quad \text{and}\quad d_2 = (1,0,0),
$
with the output function \eqref{eq:output}, while the sensor on agent $i$ only measures the first direction.
The measurements for each direction are corrupted with additive Gaussian noise sampled from $\GP(0, \text{diag}(0.2^2, 0.1^2, 0.3^2))$.
In this experiment, we design the agent $i$ to use only the directional measurement of $d_1$, while agent $j$ has access to measurements of both directions.
In consequence, without the relative measurements from agent $j$, the state of the ego agent $i$ is unobservable --- a single directional measurement is insufficient to determine the full attitude of the vehicle.  
In this way, the experiment emphasises the role of the shared information in the filter and exacerbates errors due to information incest.
Both agents receive directional measurements at 20Hz.

Agent $j$ makes a relative attitude measurement of agent $i$ at 1Hz, which is corrupted by Gaussian noise $\GP(0,Q_j)$.
The non-homogeneous noise covariance is given by $Q_j = \text{diag}(0.5^2$, $0.3^2$,  $0.2^2)$.
We run separate simulations with both of the measurement noise models considered in \eqref{eq:measurement_1} and \eqref{eq:measurementmodel_2}.

For comparison, we implement two filters on agent $i$ using different measurements, aside from the proposed algorithm. 
One filter only uses the extrinsic directional measurements and disregards relative inter-agent measurements.
The second filter uses both extrinsic and relative measurements, however, it only does a naive fusion in the logarithmic coordinates, as in \cite{wolfe2011bayesian}.

\vspace{-1mm}
\subsection{Result}
\vspace{-1mm}
\begin{figure}[htb]
    \centering
    \includegraphics[width=0.9\linewidth]{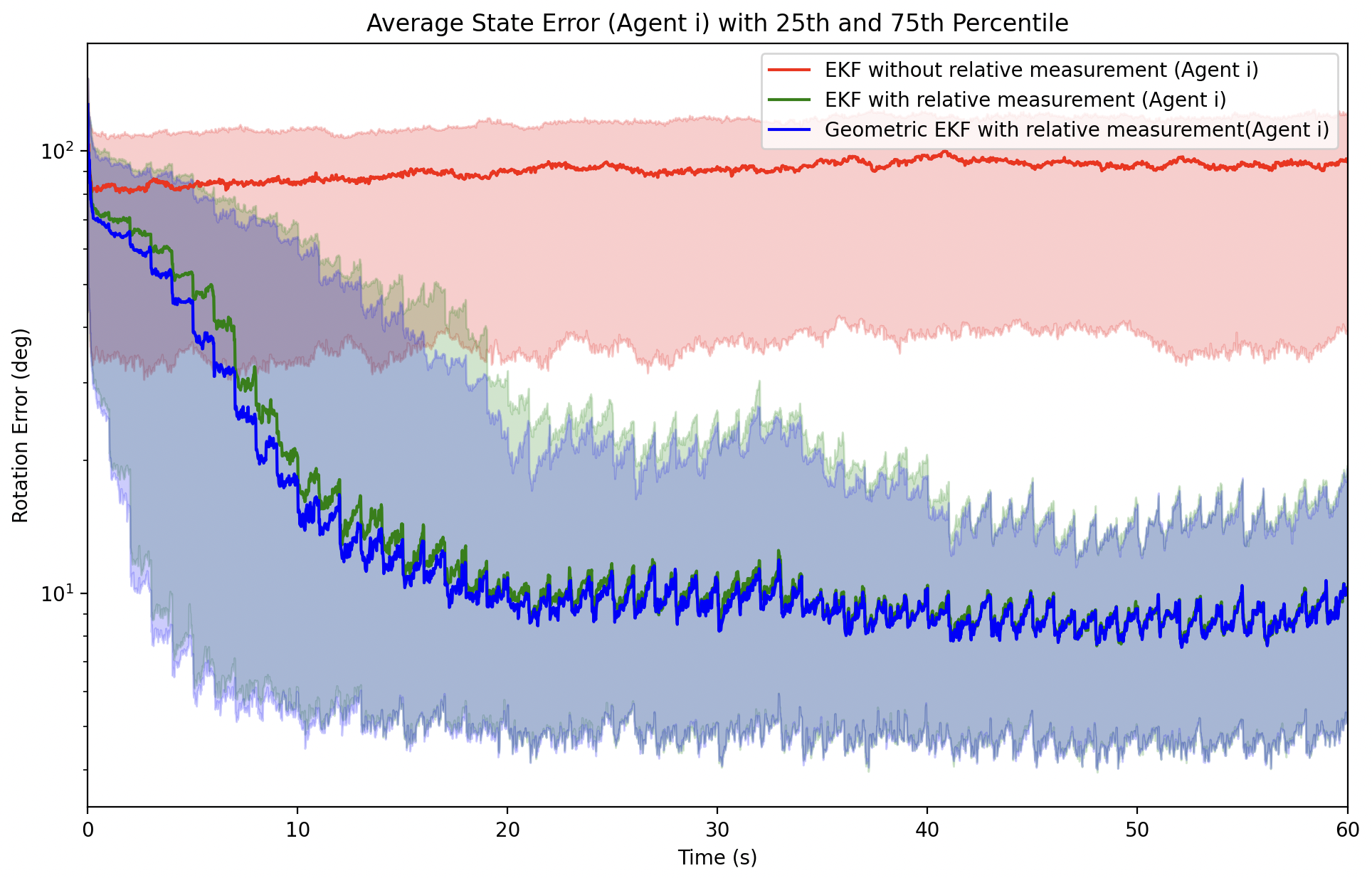}
    \vspace{-3mm}
    \caption{Direct physical measurement of relative state \eqref{eq:measurement_1}.  Mean rotation error ($e:=\arccos((\tr(R^{-1}\hat{R})-1)/2)$) with a shaded area representing the 25th and 75th percentiles.}
    \label{fig:result_model_1}
\end{figure}

Fig \ref{fig:result_model_1} and \ref{fig:result_model_2} demonstrate the performance of the proposed algorithm (in blue) compared with the EKF using only directional measurement (in red) and the EKF using a naive fusion scheme (in green).
Note that the noise parameters were chosen to demonstrate the relative advantages of the geometric modifications to be proposed.
That is, while we found that the proposed method outperformed the alternatives regardless of the noise model, the particular parameters used to generate the plots shown here were chosen to emphasise the performance advantage.
In particular, the measurement errors are larger than would be desirable in a real application, leading to relatively large attitude error in the plots.

\begin{figure}[htb]
    \centering
    \includegraphics[width=0.9\linewidth]{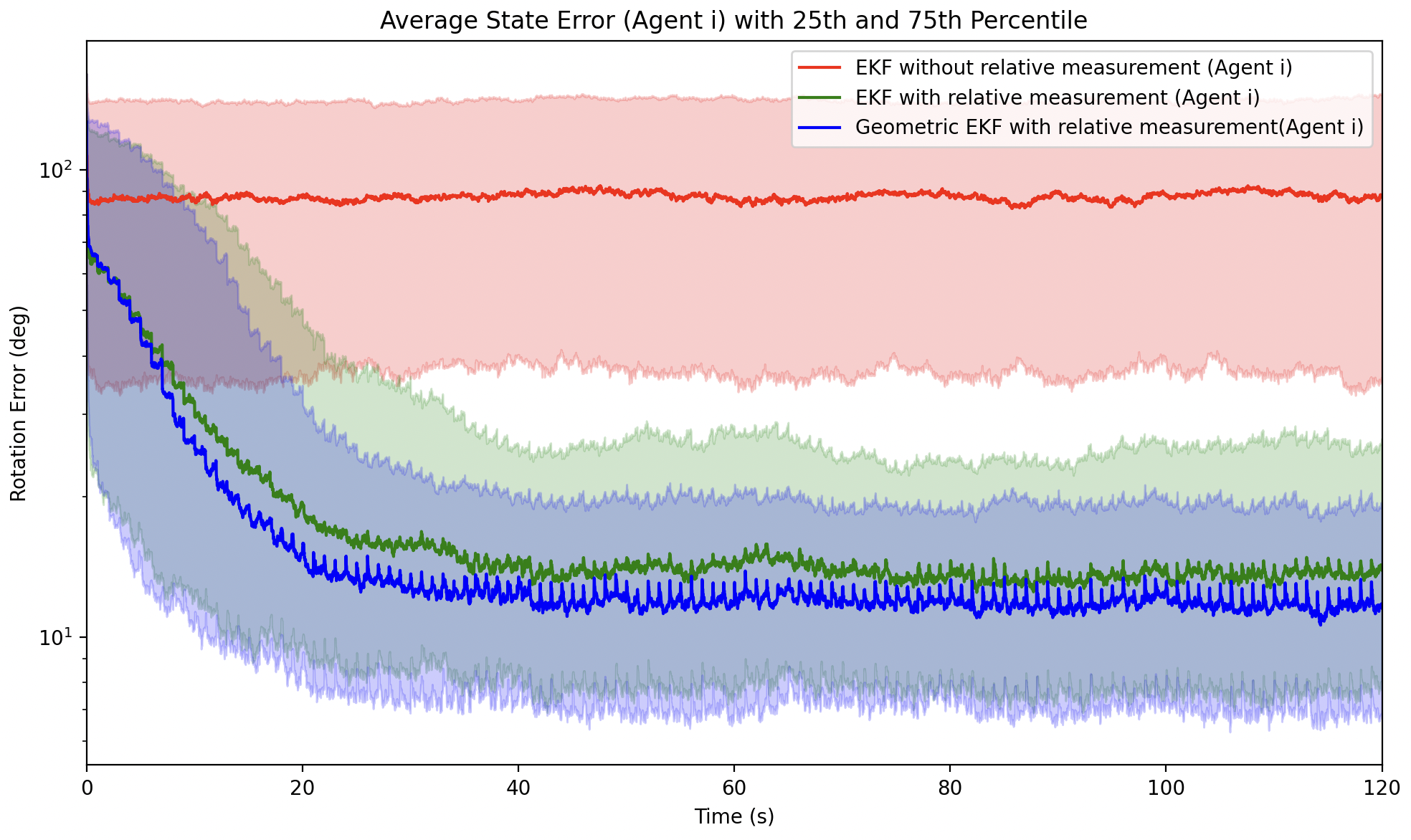}
    \vspace{-3mm}
        \caption{Relative angle measurement \eqref{eq:measurementmodel_2}.  
        }
    \label{fig:result_model_2}
\end{figure}

As expected, the local filter solution with only a single direction measurement available does not converge as the system is unobservable, as shown in both Fig \ref{fig:result_model_1} and \ref{fig:result_model_2}. 
Fig \ref{fig:result_model_1} shows the algorithm performance in the case where the direct physical measurement model \eqref{eq:measurement_1} is used.
It demonstrates an advantage in the proposed fusion algorithm during the transient of the algorithm but has similar asymptotic performance to the naive filter. 
This is to be expected, since the geometric modification in the data fusion step corrects for the difference between the filter estimate and relative state measurement coordinates.  
As the filter converges this difference becomes negligible and the benefit of the geometric correction is lost. 

This is not the case for measurement model \eqref{eq:measurementmodel_2}, as shown in Fig \ref{fig:result_model_2}, since the relative state measurement $\idx{z}{j}{}{i}$ is taken in coordinates around agent $j$. 
The relative state between the two agents $i$ and $j$ does not converge to the identity in general, and the geometric correction to compensate for the coordinate representation remains critical. 

\section{Conclusion}
In this paper, we propose a collaborative attitude estimation problem where agents run local filter-based algorithms which fuse the estimates and relative measurements communicated by neighboring agents.
We utilize the concept of concentrated Gaussians on $\SO(3)$ and exploit the geometric properties of the underlying space.
The proposed algorithm combines an EKF running locally on the agent with the CCE fusion of relative state measurements. 
We provide geometric corrections in the algorithm to incorporate Lie group geometric insights that improve filter performance. 
The simulation results demonstrate the convergence of the proposed fusion method, and show the improvements gained from applying the correct geometric modifications with different measurement noise models.

\bibliography{reference}
\bibliographystyle{IEEEtran}
\end{document}

%% file: preamble_compatible.tex
\usepackage{xcolor}  

\usepackage{amssymb}

\usepackage{amsmath}
\usepackage{mathdots}

\usepackage{mathtools}

\usepackage{tensor}

\usepackage{urwchancal}
\DeclareFontFamily{OT1}{pzc}{}
\DeclareFontShape{OT1}{pzc}{m}{it}{<-> s * [1.10] pzcmi7t}{}
\DeclareMathAlphabet{\mathpzc}{OT1}{pzc}{m}{it}

\usepackage{graphicx}
\usepackage{ifpdf}
\ifpdf
\usepackage{epstopdf}
\PrependGraphicsExtensions{.svg}
\fi

\newtheorem{theorem}{Theorem}[section]
\newtheorem{lemma}[theorem]{Lemma}
\newtheorem{corollary}[theorem]{Corollary}


%% file: notation_defs.tex
\providecommand{\Order}{\mathbb{O}}

\providecommand{\R}{\mathbb{R}}
\providecommand{\E}{\mathbb{E}} 

\providecommand{\SO}{\mathbf{SO}}

\providecommand{\grpG}{\mathbf{G}}

\providecommand{\so}{\mathfrak{so}}

\providecommand{\Sph}{\mathrm{S}}

\providecommand{\tT}{\mathrm{T}} 

\providecommand{\Id}{I} 

\providecommand{\GP}{\mathbf{N}} 

\providecommand{\tL}{\mathrm{L}} 

\DeclareMathOperator{\tr}{tr}

\DeclareMathOperator{\Ad}{Ad}
\DeclareMathOperator{\ad}{ad}

\DeclareMathOperator*{\argmin}{argmin}

\providecommand{\id}{\mathrm{id}} 

\providecommand{\tD}{\mathrm{D}}
\providecommand{\tL}{\mathrm{L}}

\usepackage{accents}
\usepackage{mathtools}
\makeatletter
\providecommand{\scirc}{%
    \hbox{\fontfamily{\rmdefault}\fontsize{0.4\dimexpr(\f@size pt)}{0}\selectfont{\raisebox{-0.52ex}[0ex][-0.52ex]{$\circ$}}}}

\makeatother

\makeatletter
\providecommand{\ucirc}{%
    \hbox{\fontfamily{\rmdefault}\fontsize{0.4\dimexpr(\f@size pt)}{0}\selectfont{\raisebox{0.0ex}[0ex][-0.52ex]{$\circ$}}}}

\makeatother

\mathchardef\mhyphen="2D

\providecommand{\idx}[4]{\tensor*[_{#3}^{#2}]{#1}{_{#4}}}

\providecommand{\Order}{\mathbf{O}} 

